\documentclass[11pt]{llncs}
\usepackage{graphicx,float}
\usepackage{amsmath,amssymb}
\usepackage{url}
\usepackage{tikz}
\usepackage{fullpage}
\usepackage{enumitem}
\usepackage[bookmarks,bookmarksdepth=2]{hyperref}
\usepackage{xcolor}
\usepackage{breakcites}
\usepackage{times}

\newcommand{\encode}{\mathsf{enc}}

\floatstyle{boxed}
\restylefloat{figure}

\title{Zero-Knowledge Password Policy Check from Lattices}
\author{Khoa Nguyen, Benjamin Hong Meng Tan, Huaxiong Wang}
\institute{
Division of Mathematical Sciences, \\
School of Physical and Mathematical Sciences,\\
Nanyang Technological University, Singapore.\\
\texttt{\{khoantt,tanh0199,hxwang\}@ntu.edu.sg}
}

\begin{document}

\maketitle

\begin{abstract}
Passwords are ubiquitous and most commonly used to authenticate users when logging into online services.
Using high entropy passwords is critical to prevent unauthorized access and password policies emerged to enforce this requirement on passwords.
However, with current methods of password storage, poor practices and server breaches have leaked many passwords to the public.
To protect one's sensitive information in case of such events, passwords should be hidden from servers.
Verifier-based password authenticated key exchange, proposed by Bellovin and Merrit (IEEE S\&P, 1992), allows authenticated secure channels to be established with a hash of a password (verifier).
Unfortunately, this restricts password policies as passwords cannot be checked from their verifier.
To address this issue, Kiefer and Manulis (ESORICS 2014) proposed zero-knowledge password policy check (ZKPPC).
A ZKPPC protocol allows users to prove in zero knowledge that a hash of the user's password satisfies the password policy required by the server.
Unfortunately, their proposal is not quantum resistant with the use of discrete logarithm-based cryptographic tools and there are currently no other viable alternatives.
In this work, we construct the first post-quantum ZKPPC using lattice-based tools.
To this end, we introduce a new randomised password hashing scheme for ASCII-based passwords and design an accompanying zero-knowledge protocol for policy compliance.
Interestingly, our proposal does not follow the framework established by Kiefer and Manulis and offers an alternate construction without homomorphic commitments.
Although our protocol is not ready to be used in practice, we think it is an important first step towards a quantum-resistant privacy-preserving password-based authentication and key exchange system.
\end{abstract}

\section{Introduction}\label{section:intro}
One of the most common methods of user authentication is passwords when logging in to online services.
So, it is very important that passwords in use have sufficient entropy and hard to guess for security.
Password policies was introduced to guide users into choosing suitable passwords that are harder to guess.
Ur et al.~\cite{UKK+12} discovered that users are more likely to choose easily guessable passwords in the absence of a password policy.
Examining the password policies of over 70 web-sites, Flor\^encio and Herley~\cite{FH10} found that most require passwords with characters from at least one of four sets, digits, symbols, lowercase and uppercase letters and a minimum password length.
Hence, it is reasonable to focus on password policies with a minimum password length, sets of valid characters and maybe constraints on the diversity of characters used.

Even with strong passwords and good policies, nothing can prevent leaks if servers do not properly store passwords.
Improperly stored passwords can cause serious problems, as seen by hacks on LinkedIn~\cite{LI12} and Yahoo~\cite{Ya16} and the web-site ``Have I Been Pwned?"~\cite{pwned}.
Sadly, such poor practices are not uncommon: many popular web-sites were discovered by Baumann et al.~\cite{BLL15} to store password information in plaintext.

If servers cannot be trusted, then no password information should be stored there at all.
Thus, protocols that do not require storing secret user information at external servers become necessary.
However, even with secret passwords, password policies are important to enforce a base level of security against dictionary attacks, leaving a dilemma: how do users prove compliance of their password without revealing anything?



Kiefer and Manulis~\cite{KM14} showed how to address this problem with zero knowledge password policy check~(ZKPPC).
It enables blind registration: users register a password with a server and prove password policy conformance without revealing anything about their passwords, thereby solving the dilemma.
With ZKPPC, some randomised password verification information is stored at the server and it does not leak information about the password, protecting against server compromises.
Furthermore, ZKPPC allows a user to prove, without revealing any information, that the password conforms to the server's policy.
Blind registration can be coupled with a verifier-based password-based authenticated key exchange~(VPAKE) protocol to achieve a complete system for privacy-preserving password-based registration, authentication and key exchange.
Password-based authenticated key exchange~(PAKE)~\cite{BM92,BPR00,GL03,KV09,DF11,BBC+13} is a protocol that allows users to simultaneously authenticate themselves using passwords and perform key exchange.
However, these protocols store passwords on the server and thus, users have to trust the security of the server's password storage and may be vulnerable to password leakages in the event of server compromise.
Verifier-based PAKE~\cite{BM93,BPR00,GMR06,BP13} extends PAKE to limit the damage caused by information leakage by storing a verifier instead.
Verifiers are a means to check that users supplied the correct passwords and are usually hash values of passwords with a salt, which makes it hard to extract the passwords from verifiers.

A ZKPPC protocol allows users to prove that their password, committed in the verifier, satisfies some password policy.
VPAKE can then be used to securely authenticate and establish keys whenever communication is required.
Together, the password is never revealed, greatly increasing the user security over current standards. Passwords are harder to guess and no longer easily compromised by server breaches.

Kiefer and Manulis~\cite{KM14} proposed a generic construction of ZKPPC using homomorphic commitments and set membership proofs.
In the same work, a concrete ZKPPC protocol was constructed using Pedersen commitments~\cite{Ped91}, whose security is based on the hardness of the discrete logarithm problem.
As such, it is vulnerable to attacks from quantum adversaries due to Shor's algorithm~\cite{Sho99} which solves the discrete logarithm problem in quantum polynomial time.
With NIST issuing a call for proposals to standardize quantum resistant cryptography~\cite{NIST16}, it is clear that we need to prepare cryptographic schemes and protocols that are quantum resistant, in case a sufficiently powerful quantum computer is realized.
As there is currently no proposal of ZKPPC protocol that has the potential to be quantum resistant, it is an interesting open problem to construct one.


\medskip\noindent
\textsc{Our contributions and techniques. }
In this work, inspired by the attractiveness of ZKPPC protocols and the emergence of lattice-based cryptography as a strong quantum resistant candidate, we aim to construct a ZKPPC protocol from lattices. Our contribution is two-fold. We first design a randomised password hashing scheme
based on the hardness of the Short Integer Solution (\textsf{SIS}) problem. We then construct a  \textsf{SIS}-based statistical zero-knowledge argument of knowledge, which allows the client to  convince the server that his secret password, committed in a given hash value, satisfies the server's policy. This yields the first ZKPPC protocol that still resists against quantum computers.

Our first technical challenge is to derive a password encoding mechanism that operates securely and interacts smoothly with available lattice-based cryptographic tools. In the discrete log setting considered in~\cite{KM14}, passwords are mapped to large integers and then encoded as elements in a group of large order. Unfortunately, this does not translate well to the lattice setting as working with large-norm objects usually makes the construction less secure and less efficient. Therefore, a different method, which encodes passwords as small-norm objects, is desirable. To this end,
we define a password encoding mechanism, that maps a password consisting of $t$ characters to a binary vector of length $8t$, where each of the $t$ blocks is the $8$-bit representation of the ASCII value of the corresponding password character. To increase its entropy, we further shuffle the arrangement of those blocks using a random permutation, and then commit to the permuted vector as well as a binary encoding of the permutation via the \textsf{SIS}-based commitment scheme proposed by Kawachi, Tanaka and Xagawa~\cite{AC:KawTanXag08}.  This commitment value is then viewed as the randomised hash value of the password.

The next technical challenge is to prove in zero-knowledge that the committed password satisfies a policy of the form $f = \big((k_D, k_U, k_L, k_S), n_{\mathtt{min}}, n_{\mathtt{max}}\big)$, which demands that the password have length at least $n_{\sf min}$ and at most $n_{\sf max}$, and contain at least $k_D$ digits, $k_S$ symbols, $k_L$ lower-case and $k_U$ upper-case letters.
To this end, we will have to prove, for instance, that a committed length-$8$ block-vector belongs to the set of vectors encoding all $10$ digits. We thus need a lattice-based sub-protocol for proving set membership. In the lattice-based world, a set membership argument system with logarithmic complexity in the cardinality of the set was proposed in~\cite{EC:LLNW16}, exploiting Stern-like protocols~\cite{Stern96} and Merkle hash trees. However, the asymptotic efficiency does not come to the front  when the underlying set has small, constant size. Here, we employ a different approach, which has linear complexity but is technically simpler and practically more efficient, based on the extend-then-permute technique for Stern's protocol, suggested by Ling et al.~\cite{PKC:LNSW13}. Finally, we use a general framework for Stern-like protocols, put forward by Libert et al.~\cite{AC:LLMNW16b}, to combine all of our sub-protocols for set membership and obtain a ZKPPC protocol.

From a practical point of view, our lattice-based ZKPPC protocol is not yet ready to be used: for a typical setting of parameters, an execution with soundness error $2^{-30}$ has communication cost around $900$ KB. We, however, believe that there are much room for improvement and view this work as the first step in designing post-quantum privacy-preserving password-based authentication and key exchange systems.

\medskip\noindent
\textsc{Related work. }
The only construction of ZKPPC was proposed by Kiefer and Manulis~\cite{KM14} using Pedersen commitments~\cite{Ped91} and a randomised password hashing scheme introduced in the same work.
It commits each character individually and uses set membership proofs to prove compliance of the entire password to a password policy.
The password hash is the sum of the committed characters and thus is linked to the set membership proofs through the homomorphic property of the commitments used.
As mentioned previously, their protocol is vulnerable to quantum adversaries and greater diversity is desirable.

To improve the efficiency of secure password registration for VPAKE~\cite{KM16a} and two server PAKE~\cite{KM16b}, Kiefer and Manulis proposed blind password registration~(BPR), a new class of cryptographic protocols that prevent password leakage from the server.
Using techinques introduced in~\cite{KM14}, Kiefer and Manulis used an efficient shuffling proof from~\cite{Furu05} to achieve $\mathcal{O}(1)$ number of set membership proofs instead of $\mathcal{O}(n_{max})$ in ZKPPC.
However, the security model considered for BPR is only suitable for honest but curious participants.
The security of ZKPPC is defined to prevent malicious users from registering bad passwords that do not conform to the given password policy.
Malicious servers also do not gain any information on the registered password from running the ZKPPC protocol.
Overall, the security model of BPR is weaker than the capabilities of ZKPPC and available instantiations are not resistant to quantum adversaries.

An alternate approach using symmetric key primitives, secure set-based policy checking~(SPC), to check password policy was proposed in~\cite{DK15}.
Policies are represented set-theoretically as monotone access structures and are mapped to linear secret sharing schemes~(LSSS).
Then, checking policy compliance corresponds to verifying whether some set is in the access structure, i.e. if the set of shares can reconstruct the secret in the LSSS.
To obtain a privacy-preserving protocol for SPC, the oblivious bloom intersection~(OBI) from~\cite{DCW13} is used.
The server constructs an LSSS that only users who fulfil the policy can obtain the right shares from the OBI and recover the secret.
Knowledge of the secret is proved with a hash of the secret with the transcript of the protocol execution and identities of the two parties, tying the protocol to the proof of knowledge.
In the proposed SPC protocol, the one-more-RSA assumption is used to guarantee that the password registration protocol is sound when used by a malicious client.
Thus, in the presence of a quantum adversary, the SPC protocol cannot be considered sound anymore.
Since the focus is on quantum resistant blind registration of passwords with malicious participants, the SPC protocol is insufficient.


\medskip
\noindent
\textsc{Zero-knowledge proofs in lattice-based cryptography.}
Early work on interactive and non-interactive proof systems~\cite{STOC:GolGol98,C:MicVad03,C:PeiVai08} for lattices exploited the geometric structure of worst-case lattice problems, and are not generally applicable in lattice-based cryptography.
More recent methods of proving relations appearing in lattice-based cryptosystems belong to the following two main families.

The first family, introduced by Lyubashevsky~\cite{PKC:Lyubashevsky08,EC:Lyubashevsky12}, uses ``rejection sampling'' techniques, and lead to relatively efficient proofs of knowledge of small secret vectors~\cite{AC:BCKLN14,C:BDLN16,CDXY17,dPL17}, and proofs of linear and multiplicative relations among committed values~\cite{ESORICS:BKLP15,EPRINT:BDOP16} in the ideal lattice setting.
However, due to the nature of ``rejection sampling'', there is a tiny probability that even an honest prover may fail to convince the verifier: i.e., protocols in this family do not have perfect completeness.
Furthermore, when proving knowledge of vectors with norm bound~$\beta$, the knowledge extractor of these protocols is only guaranteed to produce witnesses of norm bound~$g\cdot \beta$, for some factor $g>1$.
This factor, called ``soundness slack'' in~\cite{C:BDLN16,CDXY17}, may be undesirable: if an extracted witness has to be used in the security proof to solve a challenge \textsf{SIS} instance, we need the $\mathsf{SIS}_{g\cdot \beta}$ assumption, which is stronger than the~$\mathsf{SIS}_{\beta}$ assumption required by the protocol itself. 
Moreover, in some sophisticated cryptographic constructions such as the zero-knowledge password policy check protocol considered in this work, the coordinates of extracted vectors are expected to be in~$\{0,1\}$ and/or satisfy a specific pattern.
Such issues seem hard to tackle using this family of protocols.  

The second family, initiated by Ling {\it et al.}~\cite{PKC:LNSW13}, use ``decomposition-extension'' techniques in lattice-based analogues~\cite{AC:KawTanXag08} of Stern's protocol~\cite{Stern96}.
These are less efficient than those of the first family because each protocol execution admits a constant soundness error, and require repeating protocols $\omega(\log n)$ times, for a security parameter $n$, to achieve negligible soundness error.
On the upside, Stern-like protocols have perfect completeness and can handle a wide range of lattice-based relations~\cite{PKC:LinNguWan15,EC:LLNW16,CNW16,AC:LLMNW16b,AC:LLMNW16a,LNWX17}, especially when witnesses have to not only be small or binary, but also certain prescribed arrangement of coordinates.
Furthermore, unlike protocols of the first family, the extractor of Stern-like protocols can output witness vectors with the same properties expected of valid witnesses.
This feature is often crucial in the design of advanced protocols involving \textsf{ZK} proofs.
In addition, the  ``soundness slack'' issue is completely avoided, so the hardness assumptions are kept ``in place''.

\medskip\noindent
\textsc{Organization.}
In the next section, we define notations used in the paper and briefly describe the building blocks for our ZKPPC protocol.
Following that, in Section~\ref{section:construction}, we instantiate the building blocks and ZKPPC protocol with lattices-based primitives.
Finally, we summarize and conclude in Section~\ref{section:conclusion}.

\section{Preliminaries}\label{section:prelim}
    \vspace{-0.3cm}
    {\sc Notation.}
    We assume all vectors are column vectors.
    A vector $\mathbf{x}$ with coordinates $x_1, \ldots, x_m$ is written as $\mathbf{x} = (x_1, \ldots, x_m)$.
    For simplicity, concatenation of $\mathbf{x} \in \mathbb{R}^k$ and $\mathbf{y} \in \mathbb{R}^m$ is denoted with $(\mathbf{x} \| \mathbf{y}) \in \mathbb{R}^{k + m}$.
    Column-wise concatenation of matrices $\mathbf{A} \in \mathbb{R}^{n \times k}$ and $\mathbf{B} \in \mathbb{R}^{n \times m}$ is denoted by $[\mathbf{A} \hspace*{1.6pt}|\hspace*{1.6pt} \mathbf{B}] \in \mathbb{R}^{n \times (k + m)}$.
    If $S$ is a finite set, then $x \xleftarrow{\$} S$ means that $x$ is chosen uniformly at random over $S$.
    For a positive integer $n$, $[n]$ denotes the set $\{1,\ldots, n\}$ and $\mathsf{negl}(n)$ denotes a negligible function in~$n$.
    The set of all permutations of $n$ elements is denoted by $\mathcal{S}_n$.
    All logarithms are of base $2$.
    \vspace{-0.3cm}
    \subsection{Some Lattice-Based Cryptographic Ingredients}\label{subsection:lattice-tools}
    \vspace{-0.1cm}
    We first recall the average-case problem \textsf{SIS} and its link to worst-case lattice problems.
    \begin{definition}[$\mathsf{SIS}^{\infty}_{n,m,q,\beta}$ \cite{STOC:Ajtai96,STOC:GenPeiVai08}]
    Given a uniformly random matrix $\mathbf{A} \in \mathbb{Z}_q^{n \times m}$, find a non-zero vector $\mathbf{x} \in \mathbb{Z}^m$ such that~$\|\mathbf{x}\|_\infty \leq \beta$ and $\mathbf{A\cdot x=0} \bmod q.$
    \end{definition}
    The hardness of the \textsf{SIS} is guaranteed by the worst-case to average-case reduction from lattice problems.
    If $m, \beta = \mathsf{poly}(n)$, and $q > \beta\cdot\widetilde{\mathcal{O}}(\sqrt{n})$, then the $\mathsf{SIS}^{\infty}_{n,m,q,\beta}$ problem  is at least as hard as the worst-case lattice problem $\mathsf{SIVP}_\gamma$ for some $\gamma = \beta \cdot \widetilde{\mathcal{O}}(\sqrt{nm})$ (see, e.g.,~\cite{STOC:GenPeiVai08,C:MicPei13}).

    \smallskip
    \noindent
    {\bf The KTX commitment scheme.} In this work, we employ the \textsf{SIS}-based commitment scheme proposed by Kawachi, Tanaka and Xagawa~\cite{AC:KawTanXag08} (KTX).
    The scheme, with two flavours, works with lattice parameter $n$, prime modulus $q = \widetilde{\mathcal{O}}(n)$, and dimension $m = 2n\lceil\log q\rceil$. 

    In the variant that commits $t$ bits, for some fixed $t = \mathsf{poly}(n)$, the commitment~key is $(\mathbf{A}, \mathbf{B}) \xleftarrow{\$} \mathbb{Z}_q^{n \times t} \times \mathbb{Z}_q^{n \times m}$.
    To commit $\mathbf{x} \in \{0,1\}^t$, one samples randomness $\mathbf{r} \xleftarrow{\$} \{0,1\}^m$, and outputs the commitment $\mathbf{c} = \mathbf{A}\cdot \mathbf{x} + \mathbf{B}\cdot \mathbf{r} \bmod q$.
    Then, to open $\mathbf{c}$, one reveals $\mathbf{x} \in \{0,1\}^t$ and $\mathbf{r} \in \{0,1\}^m$.

    If there exists two valid openings $(\mathbf{x}_1, \mathbf{r}_1)$ and $(\mathbf{x}_2, \mathbf{r}_2)$ for the same commitment $\mathbf{c}$ and $\mathbf{x}_1 \neq \mathbf{x}_2$, then one can compute a solution to the $\mathsf{SIS}_{n,m+t, q,1}^\infty$ problem associated with the uniformly random matrix $[\mathbf{A} \mid \mathbf{B}] \in \mathbb{Z}_q^{n \times (m+t)}$.
    On the other hand, by the left-over hash lemma~\cite{STOC:Regev05}, the distribution of a valid commitment $\mathbf{c}$ is statistically close to uniform over~$\mathbb{Z}_q^n$ which implies that it is statistically hiding.

    Kawachi et al.~\cite{AC:KawTanXag08} extended the above $t$-bit commitment scheme to a string commitment scheme $\mathsf{COM}: \{0,1\}^* \times \{0,1\}^m \rightarrow \mathbb{Z}_q^n$.
    The extended scheme shares the same characteristics, statistically hiding from the parameters set and computationally binding under the $\textsf{SIS}$ assumption. 

    In this work, we use the former variant to commit to passwords, and use 
    \textsf{COM} as a building block for Stern-like zero-knowledge protocols.

    \vspace{-0.3cm}
    \subsection{Zero-Knowledge Argument Systems and Stern-like Protocols}\label{subsection:prelim:Stern}
    \vspace{-0.1cm}
    We work with statistical zero-knowledge argument systems, interactive protocols where the zero-knowledge property holds against \emph{any} cheating verifier and the soundness property holds against \emph{computationally bounded} cheating provers.
    More formally, let the set of statements-witnesses $\mathrm{R} = \{(y,w)\} \in \{0,1\}^* \times \{0,1\}^*$ be an \textsf{NP} relation.
    A two-party game $\langle \mathcal{P},\mathcal{V} \rangle$ is called an interactive argument system for the relation $\mathrm{R}$ with soundness error $e$ if two conditions hold: \vspace{-0.1cm}
    \begin{itemize}\itemsep=0.1cm
        \item {\sf Completeness.} If $(y,w) \in \mathrm{R}$ then $\mathrm{Pr}\big[\langle \mathcal{P}(y,w),\mathcal{V}(y) \rangle =1\big]=1.$
        \item {\sf Soundness.} If  $(y,w) \not \in \mathrm{R}$, then $\forall$ \textsf{PPT} $\widehat{\mathcal{P}}$: \hspace*{2.5pt}$\mathrm{Pr}[\langle \widehat{\mathcal{P}}(y,w),\mathcal{V}(y) \rangle =1] \leq e.$
    \end{itemize}\vspace{-0.1cm}
    Here and henceforth, \textsf{PPT} denotes probabilistic polynomial time.
    An argument system is statistical zero-knowledge if for any~$\widehat{\mathcal{V}}(y)$, there exists a \textsf{PPT} simulator $\mathcal{S}(y)$ which produces a simulated transcript that is statistically close to that of the real interaction between $\mathcal{P}(y,w)$ and $\widehat{\mathcal{V}}(y)$.
    A related notion is argument of knowledge, which requires the witness-extended emulation property.
    For $3$ move protocols ({\emph{i.e.}}, commitment-challenge-response), witness-extended emulation is implied by \emph{special soundness}~\cite{ACNS:Groth04}, which assumes the existence of a \textsf{PPT} extractor, taking as input a set of valid transcripts with respect to all possible values of the ``challenge'' to the same ``commitment'', and returning $w'$ such that $(y,w') \in \mathrm{R}$.

    \smallskip\noindent
    {\bf Stern-like protocols.}
    The statistical zero-knowledge arguments of knowledge presented in this work are Stern-like~\cite{Stern96} protocols.
    In particular, they are $\Sigma$-protocols
    as defined in~\cite{AC:JKPT12,AC:BCKLN14}, where~$3$ valid transcripts are needed for extraction instead of just~$2$.
    Stern's protocol was originally proposed for code-based cryptography, and adapted to lattices by Kawachi et al.~\cite{AC:KawTanXag08}.
    It was subsequently empowered by Ling et al.~\cite{PKC:LNSW13} to handle the matrix-vector relations associated with the \textsf{SIS} and inhomogeneous \textsf{SIS} problems and extended to design several lattice-based schemes: group signatures~\cite{PKC:LinNguWan15,EC:LLNW16,AC:LLMNW16b,LNWX17}, policy-based signatures~\cite{CNW16} and group encryption~\cite{AC:LLMNW16a}.

    The basic protocol has $3$ moves.
    With $\mathsf{COM}$, the
    KTX string commitment scheme~\cite{AC:KawTanXag08}, we get a statistical zero-knowledge argument of knowledge (\textsf{ZKAoK}) with perfect completeness, constant soundness error $2/3$, and communication cost $\mathcal{O}(|w|\cdot \log q)$, where $|w|$ is the total bit-size of the secret vectors.
    \medskip

    \noindent
    {\bf An abstraction of Stern's protocol.}
    We recall an abstraction of Stern's protocol, proposed in~\cite{AC:LLMNW16b}.
    Let $n, \ell, q$ be positive integers, where $\ell\geq n$, $q \geq 2$, and $\mathsf{VALID}$ be a subset of $\{0,1\}^\ell$. Suppose $\mathcal{S}$ is a finite set and every $\phi \in \mathcal{S}$ is associated with a permutation $\Gamma_\phi$ of $\ell$ elements, satisfying the following conditions:
    \vspace{-0.3cm}
    \begin{eqnarray}\label{eq:zk-equivalence}
    \begin{cases}
    \mathbf{w} \in \mathsf{VALID} \hspace*{2.5pt} \Leftrightarrow \hspace*{2.5pt} \Gamma_\phi(\mathbf{w}) \in \mathsf{VALID}, \\
    \text{If } \mathbf{w} \in \mathsf{VALID} \text{ and } \phi \text{ is uniform in } \mathcal{S}, \text{ then }  \Gamma_\phi(\mathbf{w}) \text{ is uniform in } \mathsf{VALID}.
    \end{cases}
    \end{eqnarray}
    We aim to construct a statistical \textsf{ZKAoK} for the following abstract relation: \vspace{-0.2cm}
    \begin{eqnarray*}
    \mathrm{R_{abstract}} = \big\{(\mathbf{M}, \mathbf{v}), \mathbf{w} \in \mathbb{Z}_q^{n \times \ell} \times \mathbb{Z}_q^n \times \mathsf{VALID}: \mathbf{M}\cdot \mathbf{w} = \mathbf{v} \bmod q.\big\}\vspace{-1cm}
    \end{eqnarray*}

    Stern's original protocol has $\mathsf{VALID} = \{\mathbf{w} \in \{0,1\}^\ell: \mathsf{wt}(\mathbf{w}) = k\}$, where $\mathsf{wt}(\cdot)$ denotes the Hamming weight and $k < \ell$ for some given $k$, $\mathcal{S} = \mathcal{S}_\ell$ -- the symmetric group on~$\ell$ elements, and $\Gamma_{\phi}(\mathbf{w}) = \phi(\mathbf{w})$.

    The conditions in (\ref{eq:zk-equivalence}) are key to prove that $\mathbf{w} \in \mathsf{VALID}$ in \textsf{ZK}:
    The prover $\mathcal{P}$ samples $\phi \xleftarrow{\$} \mathcal{S}$ and the verifier $\mathcal{V}$ checks that $\Gamma_\phi(\mathbf{w}) \in \mathsf{VALID}$; no additional information about $\mathbf{w}$ is revealed to $\mathcal{V}$ due to the randomness of $\phi$.
    Furthermore, to prove in \textsf{ZK} that $\mathbf{M}\cdot \mathbf{w} = \mathbf{v} \bmod q$ holds, $\mathcal{P}$ samples $\mathbf{r}_w \xleftarrow{\$} \mathbb{Z}_q^\ell$ to mask $\mathbf{w}$, and convinces $\mathcal{V}$ instead that $\mathbf{M}\cdot (\mathbf{w} + \mathbf{r}_w) = \mathbf{M}\cdot \mathbf{r}_w + \mathbf{v} \bmod q.$

    We describe the interaction between $\mathcal{P}$ and $\mathcal{V}$ in Figure~\ref{Figure:Interactive-Protocol}.
    A statistically hiding and computationally binding string commitment scheme~\textsf{COM}, e.g. the scheme in Section~\ref{subsection:lattice-tools}, is used.

    \begin{figure}[!htbp]

    \begin{enumerate}
      \item \textbf{Commitment:} Prover $\mathcal{P}$ samples $\mathbf{r}_w \xleftarrow{\$} \mathbb{Z}_q^\ell$, $\phi \xleftarrow{\$} \mathcal{S}$ and randomness $\rho_1, \rho_2, \rho_3$ for $\mathsf{COM}$. \\
    Then, a commitment $\mathrm{CMT}= \big(C_1, C_2, C_3\big)$ is sent to the verifier $\mathcal{V}$, where \vspace{-0.1cm}
        \begin{gather*}
            C_1 =  \mathsf{COM}(\phi, \mathbf{M}\cdot \mathbf{r}_w \bmod q; \rho_1), \hspace*{8.6pt}
            C_2 =  \mathsf{COM}(\Gamma_{\phi}(\mathbf{r}_w); \rho_2), \hspace*{8.6pt}
            C_3 =  \mathsf{COM}(\Gamma_{\phi}(\mathbf{w} + \mathbf{r}_w \bmod q); \rho_3).
        \end{gather*}\vspace{-0.5cm}

      \item \textbf{Challenge:} $\mathcal{V}$ sends a challenge $Ch \xleftarrow{\$} \{1,2,3\}$ to $\mathcal{P}$.
      \smallskip
      \item \textbf{Response:} Based on $Ch$, $\mathcal{P}$ sends $\mathrm{RSP}$ computed as follows:
      \smallskip
    \begin{itemize}\itemsep=0.1cm
    \item $Ch = 1$: Let $\mathbf{t}_{w} = \Gamma_{\phi}(\mathbf{w})$, $\mathbf{t}_{r} = \Gamma_{\phi}(\mathbf{r}_w)$, and $\mathrm{RSP} = (\mathbf{t}_w, \mathbf{t}_r, \rho_2, \rho_3)$.
    \item $Ch = 2$: Let $\phi_2 = \phi$, $\mathbf{w}_2 = \mathbf{w} + \mathbf{r}_w \bmod q$, and
        $\mathrm{RSP} = (\phi_2, \mathbf{w}_2, \rho_1, \rho_3)$.
    \item $Ch = 3$: Let $\phi_3 = \phi$, $\mathbf{w}_3 = \mathbf{r}_w$, and
     $\mathrm{RSP} = (\phi_3, \mathbf{w}_3, \rho_1, \rho_2)$.
    \end{itemize}
    \end{enumerate}
    \textbf{Verification:}  Receiving $\mathrm{RSP}$, $\mathcal{V}$ proceeds as follows:
              \begin{itemize}\itemsep=0.15cm
                \item $Ch = 1$: Check that $\mathbf{t}_w \in \mathsf{VALID}$, $C_2 = \mathsf{COM}(\mathbf{t}_r; \rho_2)$, ${C}_3 = \mathsf{COM}(\mathbf{t}_w + \mathbf{t}_r \bmod q; \rho_3)$.
                 \item $Ch = 2$: Check that $C_1 = \mathsf{COM}(\phi_2, \mathbf{M}\cdot \mathbf{w}_2 - \mathbf{v} \bmod q; \rho_1)$, ${C}_3 = \mathsf{COM}(\Gamma_{\phi_2}(\mathbf{w}_2); \rho_3)$.
                \item $Ch = 3$: Check that $C_1 =  \mathsf{COM}(\phi_3, \mathbf{M}\cdot \mathbf{w}_3; \rho_1), \hspace*{5pt}
            C_2 =  \mathsf{COM}(\Gamma_{\phi_3}(\mathbf{w}_3); \rho_2).$

              \end{itemize}
              In each case, $\mathcal{V}$ outputs $1$ if and only if all the conditions hold.
    \medskip
    \caption{Stern-like \textsf{ZKAoK} for the relation $\mathrm{R_{abstract}}$.}\label{Figure:Interactive-Protocol}
    \end{figure}
    The properties of the protocol are summarized in Theorem~\ref{Theorem:zk-protocol}.

    \begin{theorem}[\cite{AC:LLMNW16b}]\label{Theorem:zk-protocol}
    Assuming that $\mathsf{COM}$ is a statistically hiding and computationally binding string commitment scheme, the protocol in Figure~\ref{Figure:Interactive-Protocol} is a statistical \emph{\textsf{ZKAoK}} with perfect completeness, soundness error~$2/3$, and communication cost~$\mathcal{O}(\ell\log q)$.
    In particular:
    \begin{itemize}\itemsep=0.1cm
    \item There exists a polynomial-time simulator that, on input $(\mathbf{M}, \mathbf{v})$, outputs an accepted transcript statistically close to that produced by the real prover.
    \item There exists a polynomial-time knowledge extractor that, on input a commitment $\mathrm{CMT}$ and $3$ valid responses $(\mathrm{RSP}_1,\mathrm{RSP}_2,\mathrm{RSP}_3)$ to all $3$ possible values of the challenge $Ch$, outputs $\mathbf{w}' \in \mathsf{VALID}$ such that $\mathbf{M}\cdot \mathbf{w}' = \mathbf{v} \bmod q.$
    \end{itemize}
    \end{theorem}
    The proof of the Theorem~\ref{Theorem:zk-protocol}, which appeared in~\cite{AC:LLMNW16b}, employs standard simulation and extraction techniques for Stern-like protocols~\cite{AC:KawTanXag08,PKC:LNSW13}. The proof is provided in Appendix~\ref{appendix:zk-theorem} for the sake of completeness.

    \subsection{Password Strings and Password Policies}\label{subsection:prelim:passwords}
    Next, we present the models of password strings and policies, adapted from~\cite{KM14}.
    \smallskip

    \noindent
    \textbf{Password Strings.}
    We consider password strings $pw$ over the set of $94$ printable characters~$\Sigma_{\sf all}$ in the ASCII alphabet~$\Sigma_{\sf ASCII}$, where
    $\Sigma_{\sf all} = \Sigma_D \cup \Sigma_S \cup \Sigma_L \cup \Sigma_U \subset \Sigma_{\sf ASCII}$ is split into four disjoint subsets:
    \begin{itemize}\itemsep=0.1cm
    \item The set of $10$ digits $\Sigma_D = \{0,1,\ldots, 9\}$;
    \item The set of $32$ symbols $\Sigma_S = \big\{$ !"\#\$\%\&\rq ()*+,-./ :;$<=>$?@ [\textbackslash]\^ \_ \lq $\{|\}$ \textasciitilde~$\big\}$;
    \item The set of $26$ lower case letters, $\Sigma_L = \{a,b,\ldots, z\}$;
    \item The set of $26$ upper case letters, $\Sigma_U = \{A,B,\ldots, Z\}$.

    \end{itemize}
    We denote by $\mathtt{Dict}$ a general dictionary containing all strings that can be formed from the characters in $\Sigma_{\sf all}$.
    A password string $pw= (c_1, c_2, \ldots, c_k) \in \Sigma_{\sf all}^k \subset \mathtt{Dict}$ of length~$k$ is an ordered multi-set of characters $c_1, \ldots, c_k \in \Sigma_{\sf all}$.

    \medskip\noindent
    \textbf{Password Policies.}
    A password policy $f = \big((k_D, k_S, k_L, k_U), n_{\sf min}, n_\mathsf{max}\big)$ has six components, a minimum length $n_{\mathsf{min}}$, maximum length $n_\mathsf{max}$, and integers $k_D$, $k_S$, $k_L$ and $k_U$ that indicate the minimum number of digits, symbols, upper-case and lower-case letters, respectively, a password string must contain. We say that
    $f(pw) = \mathtt{true}$ if and only if policy $f$ is satisfied by the password string~$pw$.
    For instance,
    \begin{enumerate}\itemsep=0.1cm
    \item Policy $f = \big((1,1,1,1),8,16\big)$ indicates that password strings must be between $8$ and $16$ characters and contain at least one digit, one symbol, one lower-case and one upper-case letters.

    \item Policy $f = \big((0,2,0,1),10,14\big)$ demands that password strings must be between $10$ and $14$ characters, including at least two symbols and one upper-case letter.
    \end{enumerate}

    \begin{remark}
    In practice, password policies typically do not specify $n_\mathsf{max}$ but we can simply fix a number that upper-bounds all reasonable password lengths.
    \end{remark}

    \subsection{Randomised Password Hashing and Zero-Knowledge Password Policy Check}\label{subsection:prelim:zkppc}
    We now recall the notions of randomised password hashing and zero-knowledge password policy check. Our presentation follows~\cite{KM14,Kie16}.

    \smallskip\noindent
    \textbf{Randomised Password Hashing.} This mechanism aims to compute some password verification information that can be used later in more advanced protocols (e.g., ZKPPC and VPAKE).
    In order to prevent off-line dictionary attacks, the computation process is randomised via a pre-hash salt and hash salt. More formally, a randomised password hashing scheme $\mathcal{H}$ is
    a tuple of $5$ algorithms $\mathcal{H}  = (\mathsf{Setup}, \mathsf{PreSalt}, \mathsf{PreHash}, \mathsf{Salt}, \mathsf{Hash})$, defined as follows.
    \begin{itemize}\itemsep=0.1cm
        \item $\mathsf{Setup}(\lambda)$: On input security parameter $\lambda$, generate public parameters~$pp$, including the descriptions of the salt spaces $\mathbb{S}_P$ and $\mathbb{S}_H$.
        \item $\mathsf{PreSalt}(pp)$: On input $pp$, output a random pre-hash salt $s_P \in \mathbb{S}_P$.
        \item $\mathsf{PreHash}(pp,pw,s_P)$: On input $pp$, password $pw$ and pre-hash salt $s_P$, output a pre-hash value $P$.
        \item $\mathsf{Salt}(pp)$: On input $pp$, output a random hash salt $s_H \in \mathbb{S}_H$.
        \item $\mathsf{Hash}(pp,P,s_P,s_H)$: On input $pp$, pre-hash value $P$, pre-hash salt $s_P$ and hash salt $s_H$, output a hash value $\mathbf{h}$.
    \end{itemize}

    A secure randomised password hashing scheme $\mathcal{H}$ must satisfy $5$ requirements: \emph{pre-image resistance}, \emph{second pre-image resistance}, \emph{pre-hash entropy preservation}, \emph{entropy preservation} and
    \emph{password hiding}.

    \begin{itemize}\itemsep=0.12cm
        \item Pre-image resistance (or tight one-wayness in~\cite{BP13}):
        Let $\text{\emph{pp}} \leftarrow \mathsf{Setup}(\lambda)$ and $\mathtt{Dict}$ be a dictionary of min-entropy $\beta$.
        $Hash(\cdot)$ is a function such that $(H_i,s_{H_i}) \gets Hash(\cdot)$, where $s_{H_i} \leftarrow \mathsf{Salt}(pp)$ and $H_i \leftarrow \mathsf{Hash}(pp,P_i,s_{P_i},s_{H_i})$.
        $P_i \leftarrow \mathsf{PreHash}(pp,pw_i,s_{P_i})$ with $s_{P_i} \leftarrow$ $\mathsf{PreSalt}(pp)$ and \emph{pw}$_i$ $\xleftarrow{\$} \mathtt{Dict}$.
        $P_i$ is stored by $Hash(\cdot)$ and there is a function $\mathtt{Verify}(i,P)$ such that $\mathtt{Verify}(i,P) = 1$ if $P = P_i$.

        \smallskip
        For all \textsf{PPT} adversaries $\mathcal{A}$ running in time at most $t$, there exists a negligible function $\varepsilon(\cdot)$ such that \vspace{-0.15cm}
        \begin{align*}
            Pr[(i,P) \gets \mathcal{A}^{Hash(\cdot)});\text{ }\mathtt{Verify}(i,P) = 1] \leq \frac{\alpha t}{2^\beta t_{\mathsf{PreHash}}} + \varepsilon(\lambda)
        \end{align*}\vspace{-0.15cm}
        for small $\alpha$ and $t_{\mathsf{PreHash}}$, the running time of $\mathsf{PreHash}$. \smallskip
        \item Second pre-image resistance: For all \textsf{PPT} adversaries $\mathcal{A}$, there exists a negligible function $\varepsilon(\cdot)$ such that for $P' \leftarrow \mathcal{A}(pp,P,s_H)$,\vspace{-0.15cm}
        \begin{align*}
            Pr\big[\big(P' \neq P\big) \wedge \big(\mathsf{Hash}(pp,P,s_H) = \mathsf{Hash}(pp,P',s_H)\big)\big] \leq \varepsilon(\lambda),\vspace{-0.3cm}
        \end{align*}
        where $pp \leftarrow \mathsf{Setup}(\lambda)$, $s_P \leftarrow \mathsf{PreSalt}(pp)$, $s_H \leftarrow \mathsf{Salt}(pp)$ and $P \leftarrow \mathsf{Hash}(pp,\text{\emph{pw}},s_P)$ for any \emph{pw} $\in \mathtt{Dict}$.
        \item Pre-hash entropy preservation: For all dictionaries $\mathtt{Dict}$ that are samplable in polynomial time with min-entropy~$\beta$ and any \textsf{PPT} adversary $\mathcal{A}$, there exists a negligible function $\varepsilon(\cdot)$ such that for $(P,s_P) \leftarrow \mathcal{A}(pp)$ with $pp \leftarrow \mathsf{Setup}(\lambda)$ and random password \emph{pw} $\xleftarrow{\$} \mathtt{Dict}$,\vspace{-0.15cm}
        \begin{align*}
            Pr\big[\big(s_P \in \mathbb{S}_P\big) \wedge \big(P = \mathsf{PreHash}(pp, \text{\emph{pw}}, s_P)\big)\big] \leq 2^{-\beta} + \varepsilon(\lambda). \vspace{-0.3cm}
        \end{align*}
        \item Entropy preservation: For all min-entropy~$\beta$ polynomial-time samplable dictionaries $\mathtt{Dict}$ and any \textsf{PPT} adversary $\mathcal{A}$, there exists a negligible function $\varepsilon(\cdot)$ such that for $(H,s_P,s_H) \leftarrow \mathcal{A}(pp)$\vspace{-0.15cm}
        \begin{align*}
            Pr\big[\big(s_P\in \mathbb{S}_P\big) \wedge \big(s_H \in \mathbb{S}_H \wedge H = \mathsf{Hash}(pp,\text{\emph{pw}},s_P,s_H)\big)\big] \leq 2^{-\beta} + \varepsilon(\lambda),\vspace{-1cm}
        \end{align*}
        where $pp \leftarrow \mathsf{Setup}(\lambda)$ and \emph{pw} $\xleftarrow{\$} \mathtt{Dict}$.
        \item Password hiding: For all \textsf{PPT} adversaries $\mathcal{A} = (\mathcal{A}_1,\mathcal{A}_2)$, where $\mathcal{A}_1(pp)$ outputs two equal length passwords $\text{\emph{pw}}_0$, $\text{\emph{pw}}_1$ for $pp \leftarrow \mathsf{Setup}(\lambda)$ and $\mathcal{A}_2(H)$ outputs a bit $b'$ for $H \leftarrow \mathsf{Hash}(pp,P,s_P,s_H)$, where $s_H \leftarrow \mathsf{Salt}(\lambda)$, $s_P \leftarrow \mathsf{PreSalt}(\lambda)$ and $P \leftarrow \mathsf{PreHash}(pp,\text{\emph{pw}}_b,s_P)$ for a random bit $b \xleftarrow{\$} \{0,1\}$, there exists a negligible function $\varepsilon(\cdot)$ such that \vspace{-0.15cm}
            \begin{align*}
                \big|Pr[b = b'] - \tfrac{1}{2}\big| \leq \varepsilon(\lambda).\vspace{-0.3cm}
            \end{align*}
    \end{itemize}

    \medskip
    \noindent
    {\bf Zero-Knowledge Password Policy Check. } Let $\mathcal{H}  = (\mathsf{Setup}, \mathsf{PreSalt}, \mathsf{PreHash}, \mathsf{Salt}, \mathsf{Hash})$ be a randomised password hashing scheme.
    A password policy check (PPC) is an interactive protocol between a client and server where the password policy~$f= \big((k_D, k_S, k_L, k_U), n_{\sf min}, n_{\sf max}\big)$ of the server and public parameters $pp \leftarrow \mathsf{Setup}(\lambda)$ are used as common inputs.
    At the end of the execution, the server accepts a hash value $\mathbf{h}$ of any password $pw$ of the client's choice if and only if $f(pw) = \mathtt{true}$.
    A PPC protocol is an argument of knowledge of the password $pw$ and ssrandomness $s_P \leftarrow \mathsf{PreSalt}(pp)$, $s_H \leftarrow \mathsf{Salt}(pp)$ used for hashing.
    To prevent leaking the password to the server, one additionally requires that the protocol be zero-knowledge.

    More formally, a zero-knowledge PPC protocol is an interactive protocol between  a prover (client) and verifier (server), in which, given $(pp, f, \mathbf{h})$ the former convinces the later in zero-knowledge that the former knows $pw$ and randomness $(s_P, s_H)$ such that: \vspace{-0.2cm}
    \[
    f(pw) = \mathtt{true} \hspace*{16pt}\text{       and       }\hspace*{16pt} \mathsf{Hash}(pp,P,s_P,s_H) = \mathbf{h},\vspace{-0.2cm}
    \]
    where $P\leftarrow \mathsf{PreHash}(pp,pw,s_P)$.

\section{Our Constructions}\label{section:construction}
    To construct randomised password hashing schemes and ZKPPC protocols from concrete computational assumptions, the first challenge is to derive a password encoding mechanism that operates securely and interacts smoothly with the hashing and zero-knowledge layers.
    In the discrete log setting considered in~\cite{KM14}, passwords are mapped to large integers and then encoded as elements in a group of large order.
    Unfortunately, this does not translate well to the lattice setting as working with large-norm objects usually reduces the security and efficiency of the construction.
    Therefore, a different method, which encodes passwords as small-norm objects, is desirable.
    In this work, we will therefore use binary vectors.

    Let $\mathsf{bin}(\cdot)$ be the function that maps non-negative integers to their binary decomposition.
    For any character $c$ encoded in ASCII, let $\mathtt{ASCII}(c) \in [0, 255]$ be its code.
    Then, we define $\encode(c)$ for an ASCII encoded character $c$ and $\encode(pw)$ for some length-$t$ password $pw = (c_1, \ldots, c_t) \in \Sigma^t$ as \vspace{-0.2cm}
        \begin{align*}
    \encode(c) &= \mathsf{bin}(\mathtt{ASCII}(c)) \in \{0,1\}^8, \\
    \mathsf{encode}(pw) &= \big(\encode(c_1) \| \ldots \| \encode(c_t)\big) \in \{0,1\}^{8t}.\vspace{-0.5cm}
    \end{align*}

    \subsection{Notations, Sets and Permutations}\label{subsection:sets-permutations}
    Let $\mathfrak{m}, \mathfrak{n}$ be arbitrary positive integers. We define the following sets and permutations:

    \begin{itemize}\itemsep=1mm
    \item[$\diamond$] $\mathsf{B}_{\mathfrak{m}}^2$: the set of all vectors in $\{0,1\}^{2\mathfrak{m}}$ whose Hamming weight is exactly $\mathfrak{m}$. Note that for $\mathbf{x} \in \mathbb{Z}^{2\mathfrak{m}}$ and $\psi \in \mathcal{S}_{2\mathfrak{m}}$ the following holds:
        \begin{eqnarray}\label{eq:Bm2}
        \begin{cases}
            \mathbf{x} \in \mathsf{B}_{\mathfrak{m}}^2 \hspace*{6.8pt} \Leftrightarrow \hspace*{6.8pt} \psi(\mathbf{x}) \in \mathsf{B}_{\mathfrak{m}}^2; \\
            \mathbf{x} \in \mathsf{B}_{\mathfrak{m}}^2 \text{ and } \psi \xleftarrow{\$} \mathcal{S}_{2\mathfrak{m}}, \text{ then } \psi(\mathbf{x}) \text{ is uniform over } \mathsf{B}_{\mathfrak{m}}^2.
        \end{cases}
        \end{eqnarray}
    \item[$\diamond$] $T_{\psi, \mathfrak{n}}$, for $\psi \in \mathcal{S}_{\mathfrak{m}}$: the permutation that, when applied to a vector $\mathbf{v} = (\mathbf{v}_1 \| \mathbf{v}_2 \| \ldots \| \mathbf{v}_{\mathfrak{m}}) \in \mathbb{Z}^{\mathfrak{n}\mathfrak{m}}$, consisting of $\mathfrak{m}$ blocks of size $\mathfrak{n}$, re-arranges the blocks of $\mathbf{v}$ according to $\psi$, as follows,
    \[
    T_{\psi, \mathfrak{n}}(\mathbf{v}) = (\mathbf{v}_{\psi(1)} \| \mathbf{v}_{\psi(2)}\| \ldots \| \mathbf{v}_{\psi(\mathfrak{n})}).
    \]

    \end{itemize}

    \medskip
    \noindent
    For convenience, when working with password alphabet $\Sigma_{\sf all} = \Sigma_D \cup \Sigma_S \cup \Sigma_L \cup \Sigma_U$ and password policy $f = \big((k_D, k_U, k_L, k_S), n_{\mathtt{min}}, n_{\mathtt{max}}\big)$, we introduce the following notations and sets:
    \begin{itemize}\itemsep=1mm
    \item[$\diamond$] $\eta_D = |\Sigma_D|= 10$, $\eta_S = |\Sigma_S| =32$, $\eta_L = |\Sigma_L| = 26$, $\eta_U = |\Sigma_U| = 26$ and $\eta_{\sf all} = |\Sigma_{\sf all}|=94$. \smallskip
    \item[$\diamond$] For $\alpha \in \{D, S, L, U, \mathsf{all}\}$: $\mathsf{Enc}_\alpha = \{\encode(w)) \hspace*{2.6pt} \big| \hspace*{2.6pt} w \in \Sigma_\alpha\}$. \smallskip
    \item[$\diamond$] $\mathsf{SET}_\alpha$ for $\alpha \in \{D, S, L, U, \mathsf{all}\}$: the set of all vectors $\mathbf{v} = (\mathbf{v}_1 \| \ldots \| \mathbf{v}_{\eta_\alpha}) \in \{0,1\}^{8\eta_\alpha}$, such that the blocks $\mathbf{v}_1, \ldots, \mathbf{v}_{\eta_\alpha}  \in \{0,1\}^8$ are exactly the binary encodings of all characters in $\Sigma_\alpha$, i.e.,
        \[
        \big\{\mathbf{v}_1,  \ldots , \mathbf{v}_{\eta_\alpha}\big\} = \big\{\encode(w)): \hspace*{2.6pt} \hspace*{2.6pt} w \in \Sigma_\alpha\big\}.
        \]
    \item[$\diamond$] $\mathsf{SET}_{n_{\sf max}}$: the set of all vectors $\mathbf{v} = (\mathbf{v}_1 \| \ldots \| \mathbf{v}_{n_{\sf max}}) \in \{0,1\}^{n_{\sf max}\lceil\log n_{\sf max}\rceil}$, such that the blocks $\mathbf{v}_1, \ldots, \mathbf{v}_{n_{\sf max}} \allowdisplaybreaks \in \{0,1\}^{\lceil\log n_{\sf max}\rceil}$ are exactly the binary decompositions of all integers in $[n_{\sf max}]$, i.e.,
        \[
        \big\{\mathbf{v}_1,  \ldots , \mathbf{v}_{n_{\sf max}}\big\} = \big\{\mathsf{bin}(1), \ldots, \mathsf{bin}(n_{\sf max})\big\}.
        \]
    \end{itemize}
    Observe that the following properties hold.
    \begin{itemize}\itemsep=1mm
    \item[$\diamond$] For all $\alpha \in \{D, S, L, U, \mathsf{all}\}$, all $\mathbf{x} \in \mathbb{Z}^{8\eta_\alpha}$ and all $\psi \in \mathcal{S}_{\eta_\alpha}$:
        \begin{eqnarray}\label{eq:SET-DSLUall}
        \begin{cases}
            \mathbf{x} \in \mathsf{SET}_{\alpha} \hspace*{6.8pt} \Leftrightarrow \hspace*{6.8pt} T_{\psi, 8}(\mathbf{x})  \in \mathsf{SET}_{\alpha}; \\
            \mathbf{x} \in \mathsf{SET}_{\alpha} \text{ and } \psi \xleftarrow{\$} \mathcal{S}_{\eta_\alpha}, \text{ then } T_{\psi, 8}(\mathbf{x}) \text{ is uniform over } \mathsf{SET}_{\alpha}.
        \end{cases}
        \end{eqnarray}
    \item[$\diamond$] For all $\mathbf{x} \in \mathbb{Z}^{n_{\sf max}\lceil\log n_{\sf max}\rceil}$ and all $\psi \in \mathcal{S}_{n_{\sf max}}$:
        \begin{eqnarray}\label{eq:SET-nmin}
        \begin{cases}
            \mathbf{x} \in \mathsf{SET}_{n_{\sf max}} \hspace*{6.8pt} \Leftrightarrow \hspace*{6.8pt} T_{\psi, \lceil\log n_{\sf max}\rceil}(\mathbf{x}) \in \mathbf{x} \in \mathsf{SET}_{n_{\sf max}}; \\
            \mathbf{x}  \in \mathsf{SET}_{n_{\sf max}} \text{ and } \psi \xleftarrow{\$} \mathcal{S}_{n_{\sf max}}, \text{ then } T_{\psi, \lceil\log n_{\sf max}\rceil}(\mathbf{x}) \text{ is uniform over } \mathsf{SET}_{n_{\sf max}}.
        \end{cases}
        \end{eqnarray}
    \end{itemize}

    \subsection{Randomised Password Hashing from Lattices}\label{Subsection:lattice-pwhash}
    We describe our randomised password hashing scheme~$\mathcal{L}$ for passwords 
    of length between two given integers $n_{\sf min}$ and $n_{\sf max}$.
    At a high level, our scheme maps characters of the password $pw$ to binary block vectors, re-arranges them with a random permutation $\chi$, and finally computes the password hash as a KTX commitment (\cite{AC:KawTanXag08}, see also Section~\ref{subsection:lattice-tools}) to a vector storing all the information on $pw$ and $\chi$.
    The scheme works as follows,
    \begin{description}\itemsep=1mm
      \item[$\mathcal{L}.\mathsf{Setup}(\lambda)$.] On input security parameter $\lambda$, the algorithm performs the following steps: \smallskip
          \begin{enumerate}\itemsep=1mm
          \item Choose parameters $n = \mathcal{O}(\lambda)$, prime modulus $q = \widetilde{\mathcal{O}}(n)$, and dimension $m = 2n\lceil\log q\rceil$.
          \item Sample matrices $\mathbf{A} \xleftarrow{\$} \mathbb{Z}_q^{n \times (n_{\sf max}\lceil\log n_{\sf max}\rceil+ 8n_{\sf max})}$ and $\mathbf{B} \xleftarrow{\$} \mathbb{Z}_q^{n \times m}$.
          \item Let the pre-hash salt space be $\mathbb{S}_P = \mathcal{S}_{n_{\sf max}}$ - the set of all permutations of $n_{\sf max}$ elements, and hash salt space be $\mathbb{S}_H = \{0,1\}^m$.
          \item Output the public parameters
          $pp = \big(n,q,m, \mathbb{S}_P, \mathbb{S}_H, \mathbf{A}, \mathbf{B}\big)$.
          \end{enumerate}

      \item[$\mathcal{L}.\mathsf{PreSalt}(pp)$.] Sample $\chi \xleftarrow{\$} \mathcal{S}_{n_{\sf max}}$ and output $s_P = \chi$.

      \item[$\mathcal{L}.\mathsf{PreHash}(pp,pw,s_P)$.] Let $s_P = \chi \in \mathcal{S}_{n_{\sf max}}$ and $t \in [n_{\sf min}, n_{\sf max}]$ be the length of password $pw$. The pre-hash value $P$ is computed as follows. \smallskip
          \begin{enumerate}
            \item Compute $\mathsf{encode}({pw}) \in \{0,1\}^{8t}$, consisting of $t$ blocks of length $8$. \smallskip
            \item Insert $n_{\sf max} - t~$ blocks of length $8$, each one being $\mathsf{enc}({g})$ for some non-printable ASCII character $g \in \Sigma_{\sf ASCII}\setminus\Sigma_{\sf all}$, into the block-vector $\mathsf{encode}({pw})$ to get $\mathbf{e}  \in \{0,1\}^{8n_{\sf max}}$.\footnote{This hides the actual length $t$ of the password in the ZKPPC protocol in Section~\ref{Subsection:lattice-zkppc}.}
            \item Apply $T_{\chi, 8}$ to get $\mathbf{e}' = T_{\chi,8}(\mathbf{e}) \in \{0,1\}^{8n_{\sf max}}$. \smallskip
            \item Output the pre-hash value $P = \mathbf{e}'$.
          \end{enumerate}

      \item[$\mathcal{L}.\mathsf{Salt}(pp)$.] Sample $\mathbf{r} \xleftarrow{\$} \{0,1\}^m$ and output $s_H = \mathbf{r}$.

      \item[$\mathcal{L}.\mathsf{Hash}(pp,P,s_P,s_H)$.] Let $P = \mathbf{e}' \in \{0,1\}^{8n_{\sf max}}$, $s_P = \chi \in \mathcal{S}_{n_{\sf max}}$ and $s_H = \mathbf{r} \in \{0,1\}^m$. The hash value $\mathbf{h}$ is computed as follows, \smallskip
        \begin{enumerate}
      \item Express the permutation $\chi$ as $\chi = [\chi(1), \ldots, \chi(n_{\sf max})]$, where for each $i \in [n_{\sf max}]$, $\chi(i) \in [n_{\sf max}]$. Then, form
               \[
               \mathbf{e}_0 = \big(\mathsf{bin}(\chi(1)-1) \| \ldots \| \mathsf{bin}(\chi(n_{\sf max})-1)\big) \in \{0,1\}^{n_{\sf max}\lceil\log n_{\sf max}\rceil}.
               \]
      \item Form $\mathbf{x} = (\mathbf{e}_0 \| \mathbf{e}') \in \{0,1\}^{n_{\sf max}\lceil\log n_{\sf max}\rceil+ 8n_{\sf max}}$ and output $\mathbf{h} = \mathbf{A} \cdot \mathbf{x} + \mathbf{B} \cdot \mathbf{r}\in \mathbb{Z}_q^n$.
      \end{enumerate}
    \end{description}

    In the following theorem, we demonstrate that the proposed scheme satisfies the security requirements defined in Section~\ref{subsection:prelim:zkppc}.
    \begin{theorem}
    Under the $\mathsf{SIS}$ assumption, the randomised password hashing scheme, $\mathcal{L}$, described above satisfies $5$ requirements: \emph{pre-image resistance}, \emph{second pre-image resistance}, \emph{pre-hash entropy preservation}, \emph{entropy preservation} and
    \emph{password hiding}.
    \end{theorem}
    \begin{proof}
    First, we remark that, by construction, if the pre-hash salt $s_P = \chi$ is given, then we can reverse the procedure used to extend the length $t$ password by simply discarding any non-printable characters after applying the inverse of the permutation specified by $s_P$.
    Hence, if $s_P$ is hidden, then due to its randomness, the min-entropy of $P$ is larger than the min-entropy of $pw$.
    Thus, the proposed hashing scheme has the pre-hash entropy preservation and entropy preservation properties.

    Next, note that $\mathbf{h} = \mathbf{A}\cdot \mathbf{x} + \mathbf{B} \cdot \mathbf{r} \bmod q$ is a proper KTX commitment of message~$\mathbf{x}$ with randomness $\mathbf{r}$.
    Thus, from the statistical hiding property of the commitment scheme, the password hiding property holds.

    Furthermore, if one can produce distinct pre-hash values $P$, $P'$ that yield the same hash value $\mathbf{h}$, then one can use these values to break the computational binding property of the KTX commitment scheme.
    This implies that second pre-image resistance property holds under the \textsf{SIS} assumption.

    Finally, over the randomness of matrix $\mathbf{A}$, password $pw$ and pre-hash salt $s_P$, except for a negligible probability (i.e., in the event one accidentally finds a solution to the \textsf{SIS} problem associated with matrix $\mathbf{A}$), vector $\mathbf{A}\cdot \mathbf{x}$ accepts at least $2^\beta$ values in $\mathbb{Z}_q^n$, where $\beta$ is the min-entropy of the dictionary $\mathtt{Dict}$ from which $pw$ is chosen.
    Therefore, even if $\mathbf{A}\cdot \mathbf{x} = \mathbf{h} - \mathbf{B}\cdot s_H \bmod q$ is given, to find $P = \mathbf{e}'$, one has to perform $2^\beta$ invocations of $\mathsf{PreHash}$ which implies that the scheme satisfies the pre-image resistance property.
    \qed
    \end{proof}

    \subsection{Techniques for Proving Set Membership}\label{Subsection:lattice-zksetmembership}
    In our construction of ZKPPC in Section~\ref{Subsection:lattice-zkppc}, we will have to prove that a linear relation of the form
    \[
    \displaystyle\sum_i{\big(\text{public matrix } \mathbf{M}_i \big) \cdot \big(\text{binary secret vector } \mathbf{s}_i\big)} = \mathbf{h} \bmod q
    \]
    holds, where each secret vector $\mathbf{s}_i$ must be an element of a given set of relatively small cardinality, e.g., $\mathsf{Enc}_D, \mathsf{Enc}_S, \mathsf{Enc}_L, \mathsf{Enc}_U, \mathsf{Enc}_{\sf all}$. Thus, we need to design suitable sub-protocols to prove set membership.

    In the lattice-based world, a set membership argument system with logarithmic complexity in the cardinality of the set was proposed in~\cite{EC:LLNW16}, exploiting Stern-like protocols and Merkle hash trees.
    Despite its asymptotic efficiency, the actual efficiency is worse when the underlying set has small, constant size.
    To tackle the problems encountered here, we employ a different approach, which has linear complexity but is technically simpler and practically more efficient.

    Suppose we have to prove that an $\mathfrak{n}$-dimensional vector $\mathbf{s}_i$ belongs to a set of $\mathfrak{m}$ vectors $\{\mathbf{v_1}, \ldots, \mathbf{v}_{\mathfrak{m}}\}$.
    To this end, we append $\mathfrak{m}-1$ blocks to vector $\mathbf{s}_i$ to get an $\mathfrak{n}\mathfrak{m}$-dimensional vector $\mathbf{s}_i^\star$ whose $\mathfrak{m}$ blocks are exactly elements of the set $\{\mathbf{v_1}, \ldots, \mathbf{v}_{\mathfrak{m}}\}$.
    At the same time, we append $\mathfrak{n}(\mathfrak{m}-1)$ zero-columns to public matrix $\mathbf{M}_i$ to get matrix $\mathbf{M}^\star_i$ satisfying $\mathbf{M}^\star_i \cdot \mathbf{s}^\star_i = \mathbf{M}_i \cdot \mathbf{s}_i$, so that we preserve the linear equation under consideration.
    In this way, we reduce the set-membership problem to the problem of proving the well-formedness of $\mathbf{s}^\star_i$.
    The latter can be done via random permutations of blocks in the framework of Stern's protocol.
    For instance, to prove that $\mathbf{s}_i \in \mathsf{Enc}_D$, i.e., $\mathbf{s}_i$ is a correct binary encoding of a digit, we extend it to $\mathsf{s}^\star_i \in \mathsf{SET}_D$, apply a random permutation to the extended vector, and make use of the properties observed in~(\ref{eq:SET-DSLUall}).

    \subsection{Zero-Knowledge Password Policy Check Protocol}\label{Subsection:lattice-zkppc}
    We now present our construction of ZKPPC from lattices.
    Throughout, we use notations, sets and permutation techniques specified in Section~\ref{subsection:sets-permutations} to reduce the statement to be proved to an instance of the relation $\mathrm{R}_{\rm abstract}$ considered in Section~\ref{subsection:prelim:Stern}, which in turn can be handled by the Stern-like protocol of Figure~\ref{Figure:Interactive-Protocol}.

    Our protocol allows a prover~$\mathcal{P}$ to convince a verifier~$\mathcal{V}$ in ZK that $\mathcal{P}$ knows a password $pw$ that hashes to a given value with randomness $\chi, \mathbf{r}$, and satisfies some policy $f = \big((k_D, k_U, k_L, k_S), n_{\mathtt{min}}, n_{\mathtt{max}}\big)$.%
    \footnote{The construction we present considers the scenario where $k_D, k_S, k_L, k_U$ are all positive.
    Our scheme can be easily adjusted to handle the case where one or more of them are~$0$.}
    Recall that $\mathcal{V}$ demands $pw$ must have length between $n_{\sf min}$ and $n_{\sf max}$ inclusive, contain at least $k_D$ digits, $k_S$ symbols, $k_L$ lower-case and $k_U$ upper-case letters.
    For simplicity, we let $k_{\sf all}= n_{\sf min} - (k_D + k_U + k_L + k_S).$

    The common input consists of matrices $\mathbf{A} \in \mathbb{Z}_q^{n \times (n_{\sf max}\lceil\log n_{\sf max}\rceil+ 8n_{\sf max})}, \mathbf{B}\in \mathbb{Z}_q^{n \times m}$, hash value $\mathbf{h} \in \mathbb{Z}_q^n$ and extra information \vspace{-0.2cm} $$\Delta = (\delta_{D,1}, \ldots, \delta_{D, k_D}, \delta_{S,1}, \ldots, \delta_{S, k_S}, \delta_{L, 1}, \ldots, \delta_{L, k_L}, \delta_{U, 1}, \ldots, \delta_{U, k_U}, \delta_{\mathsf{all}, 1}, \ldots, \delta_{\mathsf{all}, k_{\sf all}}) \in [n_{\sf max}]^{n_{\sf min}},\vspace{-0.2cm}$$
    which indicates the positions of the blocks, inside vector $P = \mathbf{e}'$, encoding $k_D$ digits, $k_S$ symbols, $k_L$ lower-case letters, $k_U$ upper-case letters and $k_{\sf all}$ other printable characters within $pw$.
    Revealing $\Delta$ to $\mathcal{V}$ does not harm $\mathcal{P}$, since the original positions of those blocks (in vector $\mathbf{e}$) are protected by the secret permutation $\chi$.

    The prover's witness consists of vectors
    $\mathbf{x} = (\mathbf{e}_0 \| \mathbf{e}') \in \{0,1\}^{n_{\sf max}\lceil\log n_{\sf max}\rceil+ 8n_{\sf max}}$ and $\mathbf{r} \in \{0,1\}^m$ satisfying the following conditions:
    \begin{enumerate}\itemsep=1mm
    \item $\mathbf{A} \cdot \mathbf{x} + \mathbf{B} \cdot \mathbf{r}= \mathbf{h} \bmod q$;
    \item $\mathbf{e}_0 = \big(\hspace*{2.8pt}\mathsf{bin}(\chi(1)-1) \hspace*{2.8pt}\| \ldots \| \hspace*{2.8pt} \mathsf{bin}(\chi(n_{\sf max}-1))\hspace*{2.8pt}\big)$;
    \item $\mathbf{e}'$ has the form $(\mathbf{x}_1, \ldots, \mathbf{x}_{n_{\sf max}})$, where, for all $\alpha \in \{D, S, L, U, \mathsf{all}\}$ and all $i \in [k_\alpha]$, it holds that $\mathbf{x}_{\delta_{\alpha, i}} \in \mathsf{Enc}_\alpha$.
    \end{enumerate}
    We first observe that, if we express matrix $\mathbf{A}$ as $\mathbf{A}= \big[\mathbf{A}_0 \mid \mathbf{A}_1 \mid \ldots \mid \mathbf{A}_{n_{\sf max}}\big]$, where $\mathbf{A}_0 \in \mathbb{Z}_q^{n_{\sf max}\lceil\log n_{\sf max}\rceil}$ and $\mathbf{A}_1, \ldots, \mathbf{A}_{n_{\sf max}} \in \mathbb{Z}_q^{n \times 8}$, then equation $\mathbf{A} \cdot \mathbf{x} + \mathbf{B} \cdot \mathbf{r}= \mathbf{h} \bmod q$ can be equivalently written as
    \begin{eqnarray}\label{eq:big-linear-equation}
    \displaystyle{
    \mathbf{A}_0 \cdot \mathbf{e}_0 + \sum_{\alpha \in \{D, S, L, U, \mathsf{all}\}, i \in [k_\alpha]} \mathbf{A}_{\delta_{\alpha,i}}\cdot \mathbf{x}_{\delta_{\alpha,i}} + \sum_{j \in [n_{\sf max}] \setminus \Delta} \mathbf{A}_j \cdot \mathbf{x}_j + \mathbf{B}\cdot \mathbf{r} = \mathbf{h} \bmod q.
    }
    \end{eqnarray}

    \noindent Note that, we have $\mathbf{e}_0 \in \mathsf{SET}_{n_{\sf max}}$. We next transform the witness vectors $\mathbf{x}_1, \ldots, \mathbf{x}_{n_{\sf max}}, \mathbf{r}$ as follows,
    \begin{itemize}
    \item[$\diamond$] For all $\alpha \in \{D, S, L, U, \mathsf{all}\}$ and all $i \in [k_\alpha]$, to prove that $\mathbf{x}_{\delta_{\alpha, i}} \in \mathsf{Enc}_\alpha$, we append $\eta_{\alpha}-1$ suitable blocks to $\mathbf{x}_{\delta_{\alpha, i}}$ to get vector $\mathbf{x}^\star_{\delta_{\alpha, i}} \in \mathsf{SET}_\alpha$. \medskip
    \item[$\diamond$] For vectors $\{\mathbf{x}_j\}_{j \in [n_{\sf max}] \setminus \Delta}$, note that it is necessary and sufficient to prove that they are binary vectors (namely, they are encoding of characters that may or may not be printable).
    Similarly, we have to prove that $\mathbf{r}$ is a binary vector.
    To this end, we let $\mathbf{y} \in \{0,1\}^{8(n_{\sf max}- n_{\sf min})}$ be a concatenation of all $\{\mathbf{x}_j\}_{j \in [n_{\sf max}] \setminus \Delta}$ and $\mathbf{z} = (\mathbf{y} \| \mathbf{r}) \in \{0,1\}^{8(n_{\sf max}- n_{\sf min}) +m}$.
    Then, we append suitable binary entries to $\mathbf{z}$ to get $\mathbf{z}^\star \in \{0,1\}^{2(8(n_{\sf max}- n_{\sf min}) +m)}$ with Hamming weight exactly $8(n_{\sf max}- n_{\sf min}) +m$, i.e., $\mathbf{z}^\star \in \mathsf{B}^2_{8(n_{\sf max}- n_{\sf min}) + m}$.
    \end{itemize}
    Having performed the above transformations, we construct the vector $\mathbf{w} \in \{0,1\}^\ell$, where \vspace{-0.2cm}$$\ell = n_{\sf max}\lceil\log n_{\sf max}\rceil + 8(k_D\eta_D + k_S\eta_S + k_L\eta_L + k_U\eta_U) + 8k_{\sf all}\eta_{\sf all} + 2(8(n_{\sf max}- n_{\sf min}) + m),\vspace{-0.2cm}$$
    and $\mathbf{w}$ has the form:\vspace{-0.2cm}
    \begin{eqnarray}\label{eq:protocol-witness-w}
    \nonumber\mathbf{w} = \big(\hspace*{2.6pt}
    \mathbf{e}_0 \hspace*{2.6pt}\|\hspace*{2.6pt}
    \mathbf{x}^\star_{\delta_{D,1}} \hspace*{2.6pt}\|\hspace*{2.6pt}
    \ldots \|\hspace*{2.6pt}
    \mathbf{x}^\star_{\delta_{D, k_D}} &&\|\hspace*{2.6pt}
    \mathbf{x}^\star_{\delta_{S,1}} \hspace*{2.6pt}\|\hspace*{2.6pt}
    \ldots \|
    \mathbf{x}^\star_{\delta_{S, k_S}} \hspace*{2.6pt}\|\hspace*{2.6pt}
    \mathbf{x}^\star_{\delta_{L,1}} \hspace*{2.6pt}\|\hspace*{2.6pt}
    \ldots \|\hspace*{2.6pt}
    \mathbf{x}^\star_{\delta_{L, k_L}} \\
    &&\|\hspace*{2.6pt}
    \mathbf{x}^\star_{\delta_{U,1}} \hspace*{2.6pt}\|\hspace*{2.6pt}
    \ldots \|\hspace*{2.6pt}
    \mathbf{x}^\star_{\delta_{U, k_U}} \hspace*{2.6pt}\|\hspace*{2.6pt}
    \mathbf{x}^\star_{\delta_{\mathsf{all},1}} \hspace*{2.6pt}\|\hspace*{2.6pt}
    \ldots \|\hspace*{2.6pt}
    \mathbf{x}^\star_{\delta_{\mathsf{all}, k_{\sf all}}} \hspace*{2.6pt}\|\hspace*{2.6pt}
    \mathbf{z}^\star \hspace*{2.6pt}
    \big).\vspace{-0.2cm}
    \end{eqnarray}

    When performing extensions over the secret vectors, we also append zero-columns to the public matrices in equation~(\ref{eq:big-linear-equation}) so that it is preserved.
    Then, we concatenate the extended matrices to get $\mathbf{M} \in \mathbb{Z}_q^{n \times \ell}$ such that~(\ref{eq:big-linear-equation}) becomes, with $\mathbf{v} = \mathbf{h} \in \mathbb{Z}_q^n$,
    \begin{equation}\label{eq:final-Mw=v}
    \mathbf{M}\cdot \mathbf{w} = \mathbf{v} \bmod q.
    \end{equation}

    We have now established the first step towards reducing the given statement to an instance of the relation $\mathrm{R}_{\rm abstract}$ from Section~\ref{subsection:prelim:Stern}.
    Next, we will specify the set $\mathsf{VALID}$ containing the vector $\mathbf{w}$, set $\mathcal{S}$ and permutations $\{\Gamma_\phi: \phi \in \mathcal{S}\}$ such that the conditions in~(\ref{eq:zk-equivalence}) hold.

    Define $\mathsf{VALID}$ as the set of all vectors $\mathbf{w} \in \{0,1\}^\ell$ having the form~(\ref{eq:protocol-witness-w}), where
    \begin{itemize}\itemsep=1mm
    \item[$\diamond$] $\mathbf{e}_0 \in \mathsf{SET}_{n_{\sf max}}$;
    \item[$\diamond$] $\mathbf{x}^\star_{\delta_{\alpha, i}} \in \mathsf{SET}_\alpha$ for all $\alpha \in \{D, S, L, U, \mathsf{all}\}$ and all $i \in [k_\alpha]$;
    \item[$\diamond$] $\mathbf{z}^\star \in \mathsf{B}^2_{8(n_{\sf max}- n_{\sf min}) + m}$.
    \end{itemize}

    It can be seen that the vector $\mathbf{w}$ obtained above belongs to this tailored set $\mathsf{VALID}$.
    Next, let us define the set of permutations $\mathcal{S}$ as follows, \vspace{-0.2cm}
    \[
    \mathcal{S} = \mathcal{S}_{n_{\sf max}} \times \big(\mathcal{S}_{\eta_D}\big)^{k_D} \times \big(\mathcal{S}_{\eta_S}\big)^{k_S} \times \big(\mathcal{S}_{\eta_L}\big)^{k_L} \times \big(\mathcal{S}_{\eta_U}\big)^{k_U} \times \big(\mathcal{S}_{\eta_{\sf all}}\big)^{k_{\sf all}} \times \mathcal{S}_{2(8(n_{\sf max}- n_{\sf min}) + m)}.\vspace{-0.2cm}
    \]
    Then, for each element \vspace{-0.2cm}
    \[
    \phi = \big(\hspace*{1.6pt}
    \pi, \hspace*{1.6pt}
    \tau_{D,1}, \ldots, \hspace*{1.6pt}\tau_{D,k_D}, \hspace*{1.6pt}
    \tau_{S,1}, \ldots, \hspace*{1.6pt}\tau_{S,k_S},  \hspace*{1.6pt}
    \tau_{L,1}, \ldots, \hspace*{1.6pt}\tau_{L,k_L},   \hspace*{1.6pt}
    \tau_{U,1}, \ldots,\hspace*{1.6pt} \tau_{U,k_U},   \hspace*{1.6pt}
    \tau_{\mathsf{all},1}, \ldots, \hspace*{1.6pt}\tau_{\mathsf{all}, k_{\sf all}},  \hspace*{1.6pt}
    \theta  \hspace*{1.6pt}
    \big) \in \mathcal{S},\vspace{-0.2cm}
    \]
    we define the permutation $\Gamma_\phi$ that, when applied to $\mathbf{w} \in \mathbb{Z}^\ell$ of the form \vspace{-0.2cm}
    \begin{eqnarray*}
    \nonumber\mathbf{w} = \big(\hspace*{2.6pt}
    \mathbf{e}_0 \hspace*{2.6pt}\|\hspace*{2.6pt}
    \mathbf{x}^\star_{\delta_{D,1}} \hspace*{2.6pt}\|\hspace*{2.6pt}
    \ldots \|\hspace*{2.6pt}
    \mathbf{x}^\star_{\delta_{D, k_D}} &&\|\hspace*{2.6pt}
    \mathbf{x}^\star_{\delta_{S,1}} \hspace*{2.6pt}\|\hspace*{2.6pt}
    \ldots \|
    \mathbf{x}^\star_{\delta_{S, k_S}} \hspace*{2.6pt}\|\hspace*{2.6pt}
    \mathbf{x}^\star_{\delta_{L,1}} \hspace*{2.6pt}\|\hspace*{2.6pt}
    \ldots \|\hspace*{2.6pt}
    \mathbf{x}^\star_{\delta_{L, k_L}} \\
    &&\|\hspace*{2.6pt}
    \mathbf{x}^\star_{\delta_{U,1}} \hspace*{2.6pt}\|\hspace*{2.6pt}
    \ldots \|\hspace*{2.6pt}
    \mathbf{x}^\star_{\delta_{U, k_U}} \hspace*{2.6pt}\|\hspace*{2.6pt}
    \mathbf{x}^\star_{\delta_{\mathsf{all},1}} \hspace*{2.6pt}\|\hspace*{2.6pt}
    \ldots \|\hspace*{2.6pt}
    \mathbf{x}^\star_{\delta_{\mathsf{all}, k_{\sf all}}} \hspace*{2.6pt}\|\hspace*{2.6pt}
    \mathbf{z}^\star \hspace*{2.6pt}
    \big).\vspace{-0.2cm}
    \end{eqnarray*}
    where $\mathbf{e}_0 \in \mathbb{Z}^{n_{\sf max}\lceil\log n_{\sf max}\rceil}$, $\mathbf{x}^\star_{\delta_{\alpha, i}} \in \mathbb{Z}^{8\eta_\alpha}$ for all $\alpha \in \{D, S, L, U, \mathsf{all}\}$ and $i \in k_\alpha$, and $\mathbf{z}^\star \in \mathbb{Z}^{2(8(n_{\sf max}- n_{\sf min}) + m)}$, it transforms the blocks of vector $\mathbf{w}$ as follows,

    \begin{itemize}\itemsep=1mm
      \item[$\diamond$] $\mathbf{e}_0 \mapsto T_{\pi, \lceil\log n_{\sf max}\rceil}(\mathbf{e}_0)$.
      \item[$\diamond$] For all $\alpha \in \{D, S, L, U, \mathsf{all}\}$ and all $i \in k_\alpha$: \hspace*{2.8pt}
      $\mathbf{x}^\star_{\delta_{\alpha,i}} \mapsto T_{\tau_{\alpha,i}, 8}(\mathbf{x}^\star_{\delta_{\alpha,i}})$.
      \item[$\diamond$] $\mathbf{z}^\star \mapsto \theta(\mathbf{z}^\star)$.
    \end{itemize}

    Based on the properties observed in~(\ref{eq:Bm2}), (\ref{eq:SET-DSLUall}), and (\ref{eq:SET-nmin}), it can be seen that we have satisfied the conditions specified in~(\ref{eq:zk-equivalence}), namely, \vspace{-0.2cm}
    \begin{eqnarray*}
    \begin{cases} \mathbf{w} \in \mathsf{VALID} \hspace*{2.5pt} \Leftrightarrow \hspace*{2.5pt} \Gamma_\phi(\mathbf{w}) \in \mathsf{VALID}, \\
    \text{If } \mathbf{w} \in \mathsf{VALID} \text{ and } \phi \text{ is uniform in } \mathcal{S}, \text{ then }  \Gamma_\phi(\mathbf{w}) \text{ is uniform in } \mathsf{VALID}.
    \end{cases}\vspace{-0.2cm}
    \end{eqnarray*}

    Having reduced the considered statement to an instance of the relation $\mathrm{R}_{\rm abstract}$, let us now describe how our protocol
    is executed.
    The protocol uses the KTX string commitment scheme $\mathsf{COM}$, which is statistically hiding and computationally binding under the \textsf{SIS} assumption.
    Prior to the interaction, the prover $\mathcal{P}$ and verifier $\mathcal{V}$ construct the matrix $\mathbf{M}$ and vector $\mathbf{v}$ based on the common inputs $(\mathbf{A}, \mathbf{B}, \mathbf{h}, \Delta)$, while $\mathcal{P}$ builds the vector $\mathbf{w} \in \mathsf{VALID}$ from vectors $\mathbf{x}$ and $\mathbf{r}$, as discussed above.
    Then, $\mathcal{P}$ and $\mathcal{V}$ interact per Figure~\ref{Figure:Interactive-Protocol}.
    We thus obtain the following result, as a corollary of Theorem~\ref{Theorem:zk-protocol}.
    \begin{theorem}
    Under the $\mathsf{SIS}$ assumption, the protocol above is a ZKPPC protocol with respect to the randomised password hashing scheme $\mathcal{L}$ from Section~\ref{Subsection:lattice-pwhash} and policy $f = \big((k_D, k_U, k_L, k_S), n_{\mathtt{min}}, n_{\mathtt{max}}\big)$.
      The protocol is a statistical ZKAoK with perfect completeness, soundness error $2/3$ and communication cost $\mathcal{O}(\ell \log q)$.
    \end{theorem}
    \begin{proof}
    Perfect completeness, soundness error $2/3$ and communication cost  $\mathcal{O}(\ell \log q)$ of the protocol follow from the use of the abstract protocol in Figure~\ref{Figure:Interactive-Protocol}.
    For simulation, we simply run the simulator of Theorem~\ref{Theorem:zk-protocol}.

    As for knowledge extraction, we first run the knowledge extractor of Theorem~\ref{Theorem:zk-protocol} to get the vector $\mathbf{w}' \in \mathsf{VALID}$ such that $\mathbf{M}\cdot \mathbf{w}' = \mathbf{v} \bmod q.$
    Then, we ``backtrack'' the transformations to extract from $\mathbf{w}'$, vectors $\mathbf{x}' = (\mathbf{e}'_0 \| \mathbf{x}'_1 \| \ldots \| \mathbf{x}'_{n_{\sf max}}) \in \{0,1\}^{n_{\sf max}\lceil\log n_{\sf max}\rceil+ 8n_{\sf max}}$ and $\mathbf{r}' \in \{0,1\}^m$ such that
    \begin{itemize}\itemsep=1mm
    \item[$\diamond$] $\mathbf{A}\cdot \mathbf{x}' + \mathbf{B}\cdot \mathbf{r}' = \mathbf{h} \bmod q$;
    \item[$\diamond$] $\mathbf{e}'_0 \in \mathsf{SET}_{n_{\sf max}}$;
    \item[$\diamond$] For all $\alpha \in \{D, S, L, U, \mathsf{all}\}$ and all $i \in k_{\alpha}$: $\mathbf{x}'_{\delta_{\alpha,i}} \in \mathsf{Enc}_\alpha$.
    \end{itemize}
    Notice that one can recover a permutation of $n_{\sf max}$ elements from an element of $\mathsf{SET}_{n_{\sf max}}$.
    Let $\chi'$ be the permutation encoded by $\mathbf{e}'_0$.
    Then, by applying the inverse permutation $T^{-1}_{\chi', 8}$ to $(\mathbf{x}'_1 \| \ldots \| \mathbf{x}'_{n_{\sf max}})$, we recover $\mathbf{e}' \in \{0,1\}^{8n_{\sf max}}$.
    Finally, by removing potential blocks of length $8$ that correspond to encodings of non-printable ASCII characters from $\mathbf{e}'$, we obtain a vector that encodes some password string $pw'$ satisfying policy $f$.
     \qed
    \end{proof}
    \smallskip

    \noindent
    {\bf Efficiency analysis. } By inspection, we can see that, without using the big-O notation, each round of the proposed protocol has communication cost slightly larger than \vspace{-0.2cm}
    \[
    \ell\log q = \big(n_{\sf max}\lceil\log n_{\sf max}\rceil + 8(k_D\eta_D + k_S\eta_S + k_L\eta_L + k_U\eta_U) + 8k_{\sf all}\eta_{\sf all} + 2(8(n_{\sf max}- n_{\sf min}) + m)\big)\log q.\vspace{-0.2cm}
    \]

    Let us estimate the cost in practice.
    Note that the KTX commitment scheme can work with relatively small lattice parameters, e.g., $n=256$, $\log q = 10$, $m = 5120$.
    For a common password policy $f= \big((1,1,1,1), 8,16\big)$, the communication cost would be about $17$ KB.
    As each round has a soundness error of $2/3$, one may have to repeat the protocol many times in parallel to achieve a high level of confidence.
    For instance, if a soundness error of $2^{-30}$ is required, then one can repeat $52$ times for a final cost of around $900$ KB.
    In practical implementations, one can exploit various optimizations (e.g., instead of sending a random vector, one can send the PRNG seed used to generate it) to reduce the communication complexity.

\section{Conclusion and Open Questions}\label{section:conclusion}
Through the use of the KTX commitment scheme~\cite{AC:KawTanXag08} and a Stern-like zero-knowledge argument of set membership, we designed a lattice-based zero-knowledge protocol for proving that a committed/hashed password sent to the server satisfies the required password policy.
All together, we obtain the first ZKPPC that is based on the hardness of the \textsf{SIS} problem which to date remains quantum resistant.
Unfortunately, there are no viable VPAKE protocols from lattices that can be coupled with our ZKPPC protocol to construct a complete privacy-preserving password-based authentication and key exchange system.

Our proposed ZKPPC protocol can be employed to securely register chosen passwords at remote servers with the following security guarantees:
(1) Registered passwords are not disclosed to the server until used;
(2) Each registered password provably conforms to the specified password policy.
Although not being ready to be deployed in practice, we view this work as the first step in designing post-quantum privacy-preserving password-based authentication and key exchange systems.

We leave several open questions as potential future work:
(1) to construct a more practical lattice-based ZKPPC;
(2) to develop a lattice-based VPAKE; and
(3) to extend lattice-based ZKPPC to other PAKE protocols, such as two-server PAKE, where the passwords are secretly shared between two servers, of which we assume at most one to be compromisable.
The third question is similar to the one asked by Kiefer and Manulis~\cite{KM14} and as they noted, it is a challenge even in the classical discrete logarithm setting.

\medskip
\noindent
{\sc Acknowledgements. } We would like to thank the anonymous reviewers of ISC 2017 for helpful comments. The research is supported by Singapore Ministry of Education under Research Grant MOE2016-T2-2-014(S) and by NTU under Tier 1 grant RG143/14.

\appendix
\section{Proof of Theorem~\ref{Theorem:zk-protocol}}\label{appendix:zk-theorem}
\vspace{-0.2cm}
We provide the proof of Theorem~\ref{Theorem:zk-protocol} as it appears in~\cite{AC:LLMNW16b} but first restate the theorem.

\smallskip\noindent{\bf Theorem~\ref{Theorem:zk-protocol}. }
{\it The protocol in Figure~\ref{Figure:Interactive-Protocol} is a statistical \emph{\textsf{ZKAoK}} with perfect completeness, soundness error~$2/3$, and communication cost~${\mathcal{O}}(\ell\log q)$. Namely:
\begin{itemize}\itemsep=1mm
\item There exists a polynomial-time simulator that, on input $(\mathbf{M}, \mathbf{v})$, outputs an accepted transcript statistically close to that produced by the real prover.
\item There exists a polynomial-time knowledge extractor that, on input a commitment $\mathrm{CMT}$ and $3$ valid responses $(\mathrm{RSP}_1,\mathrm{RSP}_2,\mathrm{RSP}_3)$ to all $3$ possible values of the challenge $Ch$, outputs $\mathbf{w}' \in \mathsf{VALID}$ such that $\mathbf{M}\cdot \mathbf{w}' = \mathbf{v} \bmod q.$
\end{itemize}
}

\begin{proof}
Perfect completeness of the protocol can be checked: If a prover follows the protocol honestly, the verifier will always accept. 
It is also easy to see that the communication cost is bounded by ${\mathcal{O}}(\ell \log q)$.

We now prove that the protocol is a statistical zero-knowledge argument of knowledge.

\smallskip
\noindent
{\bf Zero-Knowledge Property. } We construct a \textsf{PPT} simulator $\mathsf{SIM}$ interacting with a (possibly dishonest) verifier $\widehat{\mathcal{V}}$, such that, given only the public inputs, with probability negligibly close to $2/3$, $\mathsf{SIM}$ outputs a simulated transcript that is statistically close to the one produced by the honest prover in the real interaction.

 $\mathsf{SIM}$ first chooses uniformly at random, $\overline{Ch} \in \{1,2,3\}$, its prediction of $Ch$ that $\widehat{\mathcal{V}}$ will \emph{not} choose.

\smallskip
\noindent
\textbf{Case }$\overline{Ch}=1$: Using basic linear algebra over $\mathbb{Z}_q$, $\mathsf{SIM}$ computes a vector $\mathbf{w}' \in \mathbb{Z}_q^\ell$ such that $\mathbf{M}\cdot \mathbf{w}' = \mathbf{v} \bmod q.$
Next, it samples $\mathbf{r}_w \xleftarrow{\$} \mathbb{Z}_q^\ell$, $\phi \xleftarrow{\$}  \mathcal{S}$, and randomness $\rho_1, \rho_2, \rho_3$ for $\mathsf{COM}$.
Finally, it sends the commitment $\mathrm{CMT}= \big(C'_1, C'_2, C'_3\big)$ to $\widehat{\mathcal{V}}$, where
    \begin{gather*}
        C'_1 =  \mathsf{COM}(\phi, \mathbf{M}\cdot \mathbf{r}_w; \rho_1), \quad
        C'_2 =  \mathsf{COM}(\Gamma_{\phi}(\mathbf{r}_w); \rho_2), \quad
        C'_3 =  \mathsf{COM}(\Gamma_{\phi}(\mathbf{w}' + \mathbf{r}_w); \rho_3).\vspace{-0.2cm}
    \end{gather*}
Receiving a challenge $Ch$ from $\widehat{\mathcal{V}}$, the simulator responds as follows:
\begin{itemize}\itemsep=1mm
\item If $Ch=1$: Output $\bot$ and abort.
\item If $Ch=2$: Send $\mathrm{RSP} = \big(\phi, \mathbf{w}' + \mathbf{r}_w, \rho_1, \rho_3 \big)$.
\item If $Ch=3$: Send $\mathrm{RSP} = \big(\phi, \mathbf{r}_w, \rho_1, \rho_2\big)$.
\end{itemize}

\smallskip

\noindent
\textbf{Case }$\overline{Ch}=2$: $\mathsf{SIM}$ samples $\mathbf{w}' \xleftarrow{\$} \mathsf{VALID}$, $\mathbf{r}_w \xleftarrow{\$} \mathbb{Z}_q^\ell$, $\phi \xleftarrow{\$}  \mathcal{S}$, and randomness $\rho_1, \rho_2, \rho_3$ for $\mathsf{COM}$.
Then, it sends the commitment $\mathrm{CMT}= \big(C'_1, C'_2, C'_3\big)$ to $\widehat{\mathcal{V}}$, where
    \begin{gather*}
        C'_1 =  \mathsf{COM}(\phi, \mathbf{M}\cdot \mathbf{r}_w; \rho_1), \quad
        C'_2 =  \mathsf{COM}(\Gamma_{\phi}(\mathbf{r}_w); \rho_2), \quad
        C'_3 =  \mathsf{COM}(\Gamma_{\phi}(\mathbf{w}' + \mathbf{r}_w); \rho_3).
    \end{gather*}
Receiving a challenge $Ch$ from $\widehat{\mathcal{V}}$, the simulator responds as follows:
\begin{itemize}\itemsep=1mm
\item If $Ch=1$: Send $\mathrm{RSP} = \big(\Gamma_\phi(\mathbf{w}'), \Gamma_\phi(\mathbf{r}_w), \rho_2, \rho_3\big)$.
\item If $Ch=2$: Output $\bot$ and abort.
\item If $Ch=3$: Send $\mathrm{RSP} = \big(\phi, \mathbf{r}_w, \rho_1, \rho_2\big)$.
\end{itemize}

\smallskip

\noindent
\textbf{Case }$\overline{Ch}=3$: $\mathsf{SIM}$ samples $\mathbf{w}' \xleftarrow{\$} \mathsf{VALID}$, $\mathbf{r}_w \xleftarrow{\$} \mathbb{Z}_q^\ell$, $\phi \xleftarrow{\$}  \mathcal{S}$, and randomness $\rho_1, \rho_2, \rho_3$ for $\mathsf{COM}$.
Then, it sends the commitment $\mathrm{CMT}= \big(C'_1, C'_2, C'_3\big)$ to $\widehat{\mathcal{V}}$, where
$C'_2 =  \mathsf{COM}(\Gamma_{\phi}(\mathbf{r}_w); \rho_2)$, $C'_3 =  \mathsf{COM}(\Gamma_{\phi}(\mathbf{w}' + \mathbf{r}_w); \rho_3)$ as in the previous two cases, while
    \begin{eqnarray*}
        C'_1 =  \mathsf{COM}(\phi, \mathbf{M}\cdot (\mathbf{w}'+ \mathbf{r}_w) - \mathbf{v}; \rho_1).
    \end{eqnarray*}
Receiving a challenge $Ch$ from $\widehat{\mathcal{V}}$, it responds as follows:
\begin{itemize}
  \item If $Ch=1$: Send $\mathrm{RSP}$ computed as in the case $(\overline{Ch}=2, Ch=1)$.
  \item If $Ch=2$: Send $\mathrm{RSP}$ computed as in the case $(\overline{Ch}=1, Ch=2)$.
 \item If $Ch=3$: Output $\bot$ and abort.
\end{itemize}
\smallskip

\noindent
Observe that, in each of the cases considered above, since $\mathsf{COM}$ is statistically hiding, the distribution of the commitment $\mathrm{CMT}$ and challenge~$Ch$ from~$\widehat{\mathcal{V}}$ are statistically close to those in the real interaction.
Hence, the probability that the simulator outputs~$\bot$ is negligibly close to~$1/3$.
Moreover, one can check that whenever the simulator does not halt, it will provide an accepted transcript, the distribution of which is statistically close to the prover's in the real interaction.
In other words, the constructed simulator can successfully impersonate the honest prover with probability negligibly close to~$2/3$.

\medskip

\noindent
{\bf Argument of Knowledge.} Suppose $\mathrm{RSP}_1 = (\mathbf{t}_w, \mathbf{t}_r, \rho_{2}, \rho_{3})$, $\mathrm{RSP}_2 = (\phi_2, \mathbf{w}_2, \rho_{1}, \rho_{3})$ and $\mathrm{RSP}_3 = (\phi_3, \mathbf{w}_3, \rho_{1}, \rho_{2})$ are $3$ valid responses to the same commitment $\mathrm{CMT} = (C_1, C_2, C_3)$, with respect to all $3$ possible values of the challenge.
The validity of these responses implies that:
\[
\begin{cases}
\mathbf{t}_w \in \mathsf{VALID}; \\
C_1 = \mathsf{COM}(\phi_2, \mathbf{M}\cdot \mathbf{w}_2 - \mathbf{v} \bmod q;\rho_1) = \mathsf{COM}(\phi_3, \mathbf{M}\cdot \mathbf{w}_3; \rho_1); \\
C_2 = \mathsf{COM}(\mathbf{t}_r; \rho_2) = \mathsf{COM}(\Gamma_{\phi_3}(\mathbf{w}_3); \rho_2); \\
{C}_3 = \mathsf{COM}(\mathbf{t}_w + \mathbf{t}_r \bmod q; \rho_3) = \mathsf{COM}(\Gamma_{\phi_2}(\mathbf{w}_2); \rho_3).
\end{cases}
\]
Since \textsf{COM} is computationally binding, it implies that
\begin{eqnarray}
\begin{cases}
\mathbf{t}_w \in \mathsf{VALID}; \hspace*{2.8pt} \phi_2 = \phi_3; \hspace*{2.8pt}\mathbf{t}_r = \Gamma_{\phi_3}(\mathbf{w}_3); \hspace*{2.8pt}\mathbf{t}_w + \mathbf{t}_r = \Gamma_{\phi_2}(\mathbf{w}_2) \bmod q; \\[2.5pt]
\mathbf{M}\cdot \mathbf{w}_2 - \mathbf{v} = \mathbf{M}\cdot \mathbf{w}_3 \bmod q.
\end{cases}
\end{eqnarray}
Since $\mathbf{t}_w \in \mathsf{VALID}$, if we let $\mathbf{w}' = [\Gamma_{\phi_2}]^{-1}(\mathbf{t}_w)$, then $\mathbf{w}' \in \mathsf{VALID}$. Furthermore, we have $$\Gamma_{\phi_2}(\mathbf{w}') + \Gamma_{\phi_2}(\mathbf{w}_3) = \Gamma_{\phi_2}(\mathbf{w}_2) \bmod q,$$
which means that $\mathbf{w}' + \mathbf{w}_3 = \mathbf{w}_2 \bmod q$, and
$\mathbf{M}\cdot \mathbf{w}' + \mathbf{M}\cdot \mathbf{w}_3 = \mathbf{M}\cdot \mathbf{w}_2 \bmod q$.
As a result, we have $\mathbf{M}\cdot \mathbf{w}' = \mathbf{v} \bmod q$, concluding the proof.
\qed
\end{proof}

\end{document}